\documentclass[10pt,conference]{IEEEtran}
\usepackage{graphicx, amsmath, amssymb}
\begin{document}


\title{Sum Capacity of Multi-source Linear Finite-field Relay Networks With Fading}
\author{\authorblockN{Sang-Woon Jeon and Sae-Young Chung}
\authorblockA{School of EECS, KAIST, Daejeon, Korea\\
Email: swjeon@kaist.ac.kr, sychung@ee.kaist.ac.kr} }
\maketitle


\newtheorem{definition}{Definition}
\newtheorem{theorem}{Theorem}
\newtheorem{lemma}{Lemma}
\newtheorem{example}{Example}
\newtheorem{corollary}{Corollary}
\newtheorem{proposition}{Proposition}
\newtheorem{conjecture}{Conjecture}
\newtheorem{remark}{Remark}

\def \diag{\operatornamewithlimits{diag}}
\def \min{\operatornamewithlimits{min}}
\def \max{\operatornamewithlimits{max}}
\def \log{\operatorname{log}}
\def \max{\operatorname{max}}
\def \rank{\operatorname{rank}}
\def \out{\operatorname{out}}
\def \exp{\operatorname{exp}}
\def \arg{\operatorname{arg}}
\def \E{\operatorname{E}}
\def \tr{\operatorname{tr}}
\def \SNR{\operatorname{SNR}}
\def \SINR{\operatorname{SINR}}
\def \dB{\operatorname{dB}}
\def \ln{\operatorname{ln}}
\def \th{\operatorname{th}}

\begin{abstract}
We study a fading linear finite-field relay network having multiple source-destination pairs.
Because of the interference created by different unicast sessions, the problem of finding its capacity region is in general difficult.
We observe that, since channels are time-varying, relays can deliver their received signals by waiting for appropriate channel realizations such that the destinations can decode their messages without interference.
We propose a block Markov encoding and relaying scheme that exploits such channel variations.
By deriving a general cut-set upper bound and an achievable rate region, we characterize the sum capacity for some classes of channel distributions and network topologies.
For example, when the channels are uniformly distributed, the sum capacity is given by the minimum average rank of the channel matrices constructed by all cuts  that separate the entire sources and destinations.
We also describe other cases where the capacity is characterized.
\end{abstract}

\section{Introduction} \label{sec:intro}
Characterizing the capacity region of wireless relay networks is one of the fundamental problems.
However, if the network has multiple unicast sessions, the problem of finding its capacity region becomes much more challenging since the transmission of other sessions acts as interference and, in general, the cut-set upper bound is not tight.
Even for the two-user Gaussian interference channel, an approximate capacity region was recently characterized \cite{Etkin:08}.

For wireless networks, there exist three fundamental issues, i.e., broadcast, interference, and fading.
In this paper, we consider a multi-source fading linear finite-field relay network, which captures these three key characteristics of wireless environment.
There have been related works dealing with wireless networks assuming interference-free receptions \cite{Niranjan:06,Amir:06} and assuming no broadcast nature \cite{Smith:07}.
The works in \cite{AvestimehrDiggaviTse:07,Mohajer:08} have considered deterministic relay networks and the work in \cite{Bhadra:06} has studied finite-field erasure networks.

Since the channels are time-varying, destinations can decode their messages without interference by transmitting them through a series of particular channel instances during multihop transmission.
As an example, consider the binary-field two-hop network in Fig. \ref{FIG:interference_alignment_mitigation}.
The symbol in each node denotes the transmit bit of that node, where $s_k$ denotes the information bit of the $k$-th source.
We notice the interference-free reception is possible if $\mathbf{H}_1\mathbf{H}_2=\mathbf{I}$, where $\mathbf{H}_1$ and $\mathbf{H}_2$ denote the channel instances of the first and second hop, respectively.
The works in \cite{Viveck1:08,Nazer:09,JeonITA:09} have also shown that, by using particular channel instances jointly, one can improve achievable rates of single-hop networks.
Based on this key observation, we derive an achievable rate region for general linear finite-field relay networks.
By comparing the achievable rate region with a cut-set upper bound, we characterize the sum capacity for some classes of channel distributions and network topologies.

\begin{figure}[t!]
  \begin{center}
  \scalebox{0.9}{\includegraphics{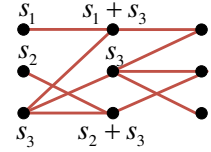}}
  \caption{Interference mitigation for two-hop networks.}
  \label{FIG:interference_alignment_mitigation}
  \end{center}
\end{figure}

\section{System Model} \label{sec:sys_model}
\subsection{Linear Finite-field Relay Networks}
We study a multi-source layered network in which the network consists of $M+1$ layers having $K_m$ nodes at the $m$-th layer, where $m\in\{1,\cdots,M+1\}$.
Let the $(k,m)$-th node denote the $k$-th node at the $m$-th layer.
The $(k,1)$-th node and the $(k,M+1)$-th node are the source and the destination of the $k$-th source-destination (S-D) pair, respectively.
Thus $K=K_1=K_{M+1}$ is the number of S-D pairs.
We define $K_{\operatorname{max}}=\max_m\{K_m\}$ and $K_{\min}=\min_m \{K_m\}$.

Consider the $m$-th hop transmission.
The $(i,m)$-th node and the $(j,m+1)$-th node become the $i$-th transmitter (Tx) and the $j$-th receiver (Rx) of the $m$-th hop, respectively, where $i\in\{1,\cdots,K_m\}$ and $j\in\{1,\cdots,K_{m+1}\}$.
Let $x_{i,m}[t]\in \mathbb{F}_2$ denote the transmit signal of the $(i,m)$-th node at time $t$ and let $y_{j,m}[t]\in \mathbb{F}_2$ denote the received signal of the $(j,m+1)$-th node at time $t$\footnote{We focus on the binary field $\mathbb{F}_2$ in this paper, but some results can be directly extended to $\mathbb{F}_q$ (see Remark \ref{RE:q_ary}).}.
Let $h_{j,i,m}[t]\in \mathbb{F}_2$ be the channel from the $(i,m)$-th node to the $(j,m+1)$-th node at time $t$.
Then the relation between the transmit and received signals is given by
\begin{equation}
y_{j,m}[t]=\sum_{i=1}^{K_m} h_{j,i,m}[t]x_{i,m}[t],
\end{equation}
where all operations are performed over $\mathbb{F}_2$.
We assume time-varying channels such that $\Pr(h_{j,i,m}[t]=1)=p_{j,i,m}$ and $h_{j,i,m}[t]$ are independent from each other with different $i$, $j$, $m$, and $t$.
Let $\mathbf{x}_m[t]$ and $\mathbf{y}_m[t]$ be the $K_m\times1$ transmit signal vector and $K_{m+1}\times1$ received signal vector at the $m$-th hop, respectively, where $\mathbf{x}_m[t]=\left[x_{1,m}[t],\cdots,x_{K_m,m}[t]\right]^T$, $\mathbf{y}_m[t]=\left[y_{1,m}[t],\cdots,y_{K_{m+1},m}[t]\right]^T$.
Thus the transmission at the $m$-th hop can be represented as
\begin{equation}
\mathbf{y}_m[t]=\mathbf{H}_m[t]\mathbf{x}_m[t],
\end{equation}
where $\mathbf{H}_m[t]$ is the $K_{m+1}\times K_m$ channel matrix of the $m$-th hop having $h_{j,i,m}[t]$ as the $(j,i)$-th element.
We assume that at time $t$ both Txs and Rxs of the $m$-th hop know $\mathbf{H}_1[t]$ through $\mathbf{H}_m[t]$.

\subsection{Problem Statement}
Consider a set of length $n$ block codes.
Let $W_k$ be the message of the $k$-th source uniformly distributed over $\{1,2,\cdots,2^{nR_k}\}$, where $R_k$ is the rate of the $k$-th source.
For simplicity, we assume that $nR_k$ is an integer.
A $\left(2^{nR_1},\cdots,2^{nR_K};n\right)$ code consists of the following encoding, relaying, and decoding functions.

\begin{itemize}
\item (Encoding) For $k\in\{1,\cdots,K\}$, the set of encoding functions of the $k$-th source is given by $\{f_{k,1,t}\}_{t=1}^n:\{1,\cdots,2^{nR_k}\}\to \mathbb{F}_2^n$
such that $x_{k,1}[t]=f_{k,1,t}(W_k)$, where $t\in\{1,\cdots,n\}$.
\item (Relaying) For $m\in\{2,\cdots,M\}$ and $k\in\{1,\cdots,K_m\}$, the set of relaying functions of the $(k,m)$-th node is given by
$\{f_{k,m,t}\}_{t=1}^n:\mathbb{F}_2^n\to \mathbb{F}_2^n$
such that $x_{k,m}[t]=f_{k,m,t}\left(y_{k,m-1}[1],\cdots, y_{k,m-1}[t-1]\right)$, where $t\in\{1,\cdots,n\}$.
\item (Decoding) For $k\in\{1,\cdots,K\}$, the decoding function of the $k$-th destination is given by
$g_k:\mathbb{F}_2^n\to\{1,\cdots,2^{nR_k}\}$
such that $\hat{W}_k=g_k\left(y_{k,M}[1],\cdots,y_{k,M}[n]\right)$.
\end{itemize}

If $M=1$, the sources transmit directly to the intended destinations without relays.
The probability of error at the $k$-th destination is given by $P^{(n)}_{e,k}=\Pr(\hat{W}_k\neq W_k)$.
A set of rates $\left(R_1,\cdots,R_K\right)$ is said to be achievable if there exists a sequence of $(2^{nR_1},\cdots,2^{nR_K};n)$ codes with $P^{(n)}_{e,k}\to 0$ as $n\to\infty$ for all $k\in\{1,\cdots,K\}$.
The achievable sum-rate is given by $R_{\operatorname{sum}}=\sum_{k=1}^{K}R_k$ and the sum capacity is the supremum of the achievable sum-rates.

\subsection{Notations}
\subsubsection{Notations for directed graphs}
The considered network can be represented as a directed graph $\mathcal{G}=(\mathcal{V},\mathcal{E})$ consisting of a vertex set $\mathcal{V}$ and a directed edge set $\mathcal{E}$.
Let $v_{k,m}$ denote the $(k,m)$-th node and $\mathcal{V}_m=\{v_{k,m}\}_{k=1}^{K_m}$ denote the set of nodes in the $m$-th layer.
Then $\mathcal{V}$ is given by $\cup_{m\in\{1,\cdots,M+1\}}\mathcal{V}_m$.
There exists a directed edge $(v_{i,m},v_{j,m+1})$ from $v_{i,m}$ to $v_{j,m+1}$ if $p_{j,i,m}>0$.
Let $\mathbf{H}_{\mathcal{V}',\mathcal{V}''}$ be the channel matrix from the nodes in $\mathcal{V}'\subseteq \mathcal{V}$ to the nodes in $\mathcal{V}''\subseteq \mathcal{V}$.

\subsubsection{Sets of channel instances and nodes}
Suppose $\bar{\mathcal{V}}'\subseteq\mathcal{V}'$, $\bar{\mathcal{V}}''\subseteq\mathcal{V}''$, and $\mathbf{G}$ is a $|\bar{\mathcal{V}}''|\times |\bar{\mathcal{V}}'|$ matrix.
We define the following set:
\begin{eqnarray}
\!\!\!\!\!\!\!\!\!\!\!&&\mathcal{H}^F_{\mathcal{V}',\mathcal{V}''}\left(\mathbf{G}, \bar{\mathcal{V}}',\bar{\mathcal{V}}''\right)=\{\mathbf{H}_{\mathcal{V}',\mathcal{V}''}\big|\mathbf{H}_{\bar{\mathcal{V}}',\bar{\mathcal{V}}''}=\mathbf{G}, \nonumber\\
\!\!\!\!\!\!\!\!\!\!\!&&{~~~~}\operatorname{rank}(\mathbf{H}_{\mathcal{V}',\mathcal{V}''})=\operatorname{rank}(\mathbf{G}),\mathbf{H}_{\mathcal{V}',\mathcal{V}''}\in \mathbb{F}_2^{|\mathcal{V}''|\times|\mathcal{V}'|}\},
\end{eqnarray}
i.e., $\mathcal{H}^F_{\mathcal{V}',\mathcal{V}''}\left(\mathbf{G}, \bar{\mathcal{V}}',\bar{\mathcal{V}}''\right)$ is the set of all instances of $\mathbf{H}_{\mathcal{V}',\mathcal{V}''}$ that contain $\mathbf{G}$ in $\mathbf{H}_{\bar{\mathcal{V}}',\bar{\mathcal{V}}''}$ and have the same rank as $\mathbf{G}$.
We further define the following sets:
\begin{eqnarray}
\!\!\!\!\!\!\!\!\!\!\!\!\!\!\!\!\!\!\!\!&&\mathcal{V}(a,b,\mathcal{V}',\mathcal{V}'')=\big\{(\bar{\mathcal{V}}',\bar{\mathcal{V}}'')\big||\bar{\mathcal{V}}'|=a,|\bar{\mathcal{V}}''|=b,\nonumber\\
\!\!\!\!\!\!\!\!\!\!\!\!\!\!\!\!\!\!\!\!&&{~~~~~~~~~~~~~~~~~~~~~~~}\bar{\mathcal{V}}'\subseteq \mathcal{V}',\bar{\mathcal{V}}''\subseteq \mathcal{V}'' \big\},\nonumber\\
\!\!\!\!\!\!\!\!\!\!\!\!\!\!\!\!\!\!\!\!&&\mathcal{V}\left(\mathbf{H}_{\mathcal{V}',\mathcal{V}''}\right)=\big\{(\bar{\mathcal{V}}',\bar{\mathcal{V}}'')\big||\bar{\mathcal{V}}'|=|\bar{\mathcal{V}}''|=\rank(\mathbf{H}_{\mathcal{V}',\mathcal{V}''}),\nonumber\\
\!\!\!\!\!\!\!\!\!\!\!\!\!\!\!\!\!\!\!\!&&{~~}\rank(\mathbf{H}_{\bar{\mathcal{V}}',\bar{\mathcal{V}}''})=\rank(\mathbf{H}_{\mathcal{V}',\mathcal{V}''}),\bar{\mathcal{V}}'\subseteq \mathcal{V}',\bar{\mathcal{V}}''\subseteq \mathcal{V}''\big\},
\end{eqnarray}
where $a$ and $b$ are positive integers satisfying $a\leq |\mathcal{V}'|$ and $b\leq |\mathcal{V}''|$.
We define $\mathcal{V}\left(\mathbf{H}_{\mathcal{V}',\mathcal{V}''}\right)=\phi$ if $\rank(\mathbf{H}_{\mathcal{V}',\mathcal{V}''})=0$.
The set $\mathcal{V}(a,b,\mathcal{V}',\mathcal{V}'')$ consists of all $(\bar{\mathcal{V}}',\bar{\mathcal{V}}'')$ such that the numbers of nodes in $\bar{\mathcal{V}}'$ and in $\bar{\mathcal{V}}''$ are equal to $a$ and $b$, respectively.
The set $\mathcal{V}\left(\mathbf{H}_{\mathcal{V}',\mathcal{V}''}\right)$ consists of all $(\bar{\mathcal{V}}',\bar{\mathcal{V}}'')$ that $\mathbf{H}_{\bar{\mathcal{V}}',\bar{\mathcal{V}}''}$ is a full rank matrix and has the same rank as $\mathbf{H}_{\mathcal{V}',\mathcal{V}''}$.

\section{Cut-set Upper Bound} \label{sec:converse}
In this section, we introduce a sum-rate upper bound, which is derived from the general cut-set upper bound in \cite{JeonITA:09}.

\begin{theorem} \label{THM:cut_set}
Suppose a linear finite-field relay network. The achievable sum-rate is upper bounded by
\begin{equation}
R_{\operatorname{sum}}\leq\min_{m\in\{1,\cdots,M\}}\mathbb{E}(\rank(\mathbf{H}_m)).
\label{EQ:converse_multi}
\end{equation}
\end{theorem}
\begin{proof}
We refer readers to \cite{JeonITA:09}.
\end{proof}
Let us define $m_0=\arg\min_m\mathbb{E}(\rank(\mathbf{H}_m))$, which is the bottleneck-hop for the entire multihop transmission\footnote{Ties are broken arbitrarily.}.

\section{Transmission Scheme} \label{sec:achievability_multi}
In this section, we propose a transmission scheme for linear finite-field relay networks when $M\geq2$.
We refer to the results in \cite{JeonITA:09} for the single-hop case.

As mentioned before, due to the time-varying nature of wireless channels, information bits can be transmitted through particular instances from $\mathbf{H}_1$ to $\mathbf{H}_M$ such that the corresponding destinations receive information bits without interference.
That is, information bits are transmitted using time indices $t_1,\cdots,t_M$ such that $\mathbf{H}_1[t_1]=\mathbf{H}_1,\cdots,\mathbf{H}_M[t_M]=\mathbf{H}_M$.
A block Markov encoding and relaying structure makes a series of pairing from $\mathbf{H}_1$ to $\mathbf{H}_M$ possible.
We first study $2$-$2$-$2$ networks and then extend the idea to general linear finite-field relay networks.

\subsection{$2$-$2$-$2$ Networks}

\begin{figure}[t!]
  \begin{center}
  \scalebox{0.72}{\includegraphics{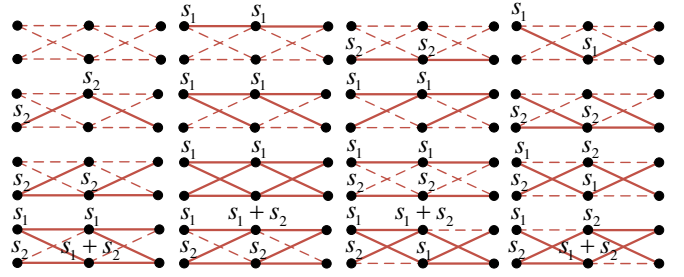}}
  \caption{Deterministic pairing of $\mathbf{H}_1$ and $\mathbf{H}_2$.}
  \label{FIG:deterministic_mapping_2hop}
  \end{center}
\end{figure}

Consider $2$-$2$-$2$ networks with $p_{j,i,m}=1/2$ for all $i$, $j$, and $m$.
There are $16$ possible instances for each $\mathbf{H}_1[t]$ and $\mathbf{H}_2[t]$.
For each time $t$, if information bits are transmitted through $\mathbf{H}_1[t]$ and $\mathbf{H}_2[t+1]$, there exist $256$ possible instances from $\mathbf{H}_1[t]$ to $\mathbf{H}_2[t+1]$ and $R_{\operatorname{sum}}=\mathbb{E}(\rank(\mathbf{H}_1\mathbf{H}_2))=\frac{177}{256}$ is achievable in this case.

However, we can get an achievable sum-rate higher than $\frac{177}{256}$ by appropriately pairing $\mathbf{H}_1$ and $\mathbf{H}_2$.
Fig. \ref{FIG:deterministic_mapping_2hop} illustrates the deterministic pairing of $\mathbf{H}_1$ and $\mathbf{H}_2$ and related encoding and relaying, where the dashed lines and the solid lines denote the corresponding channels are zeros and ones, respectively.
The symbols in the figure denote the transmit bits of the nodes and the nodes with no symbol transmit zeros, where $s_k$ denotes the information bit of the $k$-th source.
This deterministic pairing achieves $R_{\operatorname{sum}}=\mathbb{E}(\rank(\mathbf{H}_1))=\frac{21}{16}$, which coincides with the upper bound in (\ref{EQ:converse_multi}).
Thus, this simple scheme achieves the sum capacity.

Based on the deterministic pairing in Fig. \ref{FIG:deterministic_mapping_2hop}, we characterize the sum capacity for more general channel distributions.

\begin{theorem} \label{THM:2_2_2networks}
Suppose a linear finite-field relay network with $M=2$ and $K_1=K_2=K_3=2$.
Then the sum capacity is characterized for the following cases.
\begin{itemize}
\item Symmetric channel satisfying $p_{1,1,1}=p_{2,2,1}=p_{1,1,2}=p_{2,2,2}$ and $p_{2,1,1}=p_{1,2,1}=p_{2,1,2}=p_{1,2,2}$.
\item $Z$ channel satisfying $p_{2,1,1}=p_{2,1,2}=0$ with $p_{1,1,1}=p_{2,2,2}$, $p_{2,2,1}=p_{1,1,2}$ and $p_{1,2,1}=p_{1,2,2}$.
\item $\Pi$ channel satisfying $p_{1,1,1}=p_{2,2,2}=0$ with $p_{2,1,1}=p_{1,2,2}$, $p_{1,2,1}=p_{2,1,2}$ and $p_{2,2,1}=p_{1,1,2}$.
\end{itemize}
\end{theorem}
\begin{proof}
We refer readers to the full paper \cite{Jeon:09}.
\end{proof}

\subsection{General Multihop Networks}
In this subsection, we propose a transmission scheme for general linear finite-field relay networks when $M\geq 2$ and $p_{j,i,m}=p$ for all $i$, $j$, and $m$.
We assume symmetric rates for all S-D pairs, that is $R_1=\cdots=R_K$, and consider the following class of networks.
\begin{definition}
A linear finite-field relay network is said to have a minimum-dimensional bottleneck-hop $m_0$ if $K_m\geq K_{m_0}$ and $K_{m+1}\geq K_{m_0+1}$ (or $K_m\geq K_{m_0+1}$ and $K_{m+1}\geq K_{m_0}$) for all $m\in\{1,\cdots,M\}$.
\end{definition}

Notice that any networks having $K_m=K$ for all $m$ or any $2$-hop networks are included in this class of networks regardless of the value of $p$.

If a series of pairing from $\mathbf{H}_1$ to $\mathbf{H}_M$ satisfies the condition $\mathbf{H}_1\mathbf{H}_2\cdots\mathbf{H}_M=\mathbf{I}$, each destination can receive information bits without interference.
But if some instances are rank-deficient, we cannot construct such pairs by using rank-deficient instances.
Furthermore, the number of possible pairs increases exponentially as the number of nodes in a layer or the number of layers increases.
Instead, we randomize a series of pairing such that $\mathbf{H}_m$ is paired at random with one instance in a subset of $\mathbf{H}_{m+1}$'s.

\subsubsection{Block Markov encoding and relaying}
The proposed scheme divides a block into $B+M-1$ sub-blocks having length $n_B$ for each sub-block, where $n_B=\frac{n}{B+M-1}$.
Since block Markov encoding and relaying are applied over $M$ hops, the number of effective sub-blocks is equal to $B$.
Thus, the overall rate is given by $\frac{B}{B+M-1}R$, where $R$ is the symmetric rate of each sub-block.
As $n\to\infty$, the fractional rate loss $1-\frac{B}{B+M-1}$ will be negligible because we can make both $n_B$ and $B$ arbitrarily large.
For simplicity, we omit the block index in describing the proposed scheme.

\subsubsection{Balancing the average rank of each hop}
Recall that the $m_0$-th hop becomes a bottleneck for the entire multihop transmission, which can be seen from the sum-rate upper bound in (\ref{EQ:converse_multi}).
As an example, consider $3$-$2$-$2$-$3$ networks in which the second hop becomes a bottleneck.
If each source in $\mathcal{V}_1$ transmits at a rate of $\frac{1}{K}\mathbb{E}(\rank(\mathbf{H}_1))$, then it will cause an error at the second hop.
To prevent this error event, the rate of each source should be decreased to $\frac{1}{K}\mathbb{E}(\rank(\mathbf{H}_{m_0}))$.
For this reason, we select $\mathcal{V}_{m,\operatorname{tx}}[t]\subseteq\mathcal{V}_m$ and $\mathcal{V}_{m,\operatorname{rx}}[t]\subseteq\mathcal{V}_{m+1}$ randomly such that
\begin{equation}
(\mathcal{V}_{m,\operatorname{tx}}[t],\mathcal{V}_{m,\operatorname{rx}}[t])\in\mathcal{V}(K_{m_0},K_{m_0+1},\mathcal{V}_m,\mathcal{V}_{m+1})
\end{equation}
with equal probabilities (or in $\mathcal{V}(K_{m_0+1},K_{m_0},\mathcal{V}_m,\mathcal{V}_{m+1})$).
For each time $t$, only the nodes in $\mathcal{V}_{m,\operatorname{tx}}[t]$ and $\mathcal{V}_{m,\operatorname{rx}}[t]$ will become active at the $m$-th hop.
Notice that since the considered network has a minimum-dimensional bottleneck-hop, it is possible to construct such $\mathcal{V}_{m,\operatorname{tx}}[t]$ and $\mathcal{V}_{m,\operatorname{rx}}[t]$.
In the case of $3$-$2$-$2$-$3$ networks, only the nodes in $\mathcal{V}_{1,\operatorname{tx}}[t]$ and $\mathcal{V}_{1,\operatorname{rx}}[t]$ satisfying $(|\mathcal{V}_{1,\operatorname{tx}}[t]|,|\mathcal{V}_{1,\operatorname{rx}}[t]|)=(2,2)$ become active at the first hop.
The same is true for the last hop.
Whereas the whole nodes in $\mathcal{V}_2$ and $\mathcal{V}_3$ become active at the second hop, that is $\mathcal{V}_{2,\operatorname{tx}}[t]=\mathcal{V}_2$ and $\mathcal{V}_{2,\operatorname{rx}}[t]=\mathcal{V}_3$.

The following lemma shows the probability distribution of $\mathbf{H}_{\mathcal{V}_{m,\operatorname{tx}}[t],\mathcal{V}_{m,\operatorname{rx}}[t]}[t]$, which will be used to derive the achievable rate region of general multihop networks.
\begin{lemma} \label{LEM:mapping1}
Suppose a linear finite-filed relay with $p_{j,i,m}=p$ for all $i$, $j$, and $m$.
If the network has a minimum-dimensional bottleneck-hop, then
\begin{equation}
\Pr(\mathbf{H}_{\mathcal{V}_{m,\operatorname{tx}}[t],\mathcal{V}_{m,\operatorname{rx}}[t]}[t]=\mathbf{H})=p^{u}(1-p)^{K_{m_0+1}K_{m_0}-u},
\label{EQ:p_h}
\end{equation}
where $u$ is the number of zeros in $\mathbf{H}\in\mathbb{F}^{K_{m_0+1}\times K_{m_0}}_2$ or in $\mathbf{H}\in \mathbb{F}^{K_{m_0}\times K_{m_0+1}}_2$.
\end{lemma}
\begin{proof}
We refer readers to the full paper \cite{Jeon:09}.
\end{proof}

The probability distribution of $\mathbf{H}_{\mathcal{V}_{m,\operatorname{tx}}[t],\mathcal{V}_{m,\operatorname{rx}}[t]}[t]$ is the same as that of $\mathbf{H}_{m_0}[t]$, which is the channel matrix of the bottleneck-hop.
Thus if only the nodes in $\mathcal{V}_{m,\operatorname{tx}}[t]$ and $\mathcal{V}_{m,\operatorname{rx}}[t]$ are activated at the $m$-th hop, each hop can deliver information bits that are sustainable at the bottleneck-hop.

\subsubsection{Construction of transmit and receive node sets}
Because the maximum number of bits transmitted at the $m$-th hop is determined by  $\rank(\mathbf{H}_{\mathcal{V}_{m,\operatorname{tx}}[t],\mathcal{V}_{m,\operatorname{rx}}[t]}[t])$, we further select $\bar{\mathcal{V}}_{m,\operatorname{tx}}[t]\subseteq\mathcal{V}_{m,\operatorname{tx}}[t]$ and $\bar{\mathcal{V}}_{m,\operatorname{rx}}[t]\subseteq\mathcal{V}_{m,\operatorname{rx}}[t]$ randomly such that
\begin{equation}
(\bar{\mathcal{V}}_{m,\operatorname{tx}}[t],\bar{\mathcal{V}}_{m,\operatorname{rx}}[t])\in\mathcal{V}(\mathbf{H}_{\mathcal{V}_{m,\operatorname{tx}}[t],\mathcal{V}_{m,\operatorname{rx}}[t]}[t])
\end{equation}
with equal probabilities.
For each time $t$, the nodes in $\bar{\mathcal{V}}_{m,\operatorname{tx}}[t]$ transmit and the nodes in $\bar{\mathcal{V}}_{m,\operatorname{rx}}[t]$ receive through the channel $\mathbf{H}_{\bar{\mathcal{V}}_{m,\operatorname{tx}}[t],\bar{\mathcal{V}}_{m,\operatorname{rx}}[t]}[t]$ at the $m$-th hop.
Then, as we will show later, information bits can be transmitted using particular time indices $t_1,\cdots,t_M$ such that
\begin{equation}
\mathbf{H}_{\bar{\mathcal{V}}_{1,\operatorname{tx}}[t_1],\bar{\mathcal{V}}_{1,\operatorname{rx}}[t_1]}[t_1]\cdots\mathbf{H}_{\bar{\mathcal{V}}_{M,\operatorname{tx}}[t_M],\bar{\mathcal{V}}_{M,\operatorname{rx}}[t_M]}[t_M]=\mathbf{I},
\end{equation}
which guarantees interference-free reception at the destinations.
One of the simplest way is to set $\mathbf{H}_{\bar{\mathcal{V}}_{1,\operatorname{tx}}[t],\bar{\mathcal{V}}_{1,\operatorname{rx}}[t]}[t]=\cdots=\mathbf{H}_{\bar{\mathcal{V}}_{M-1,\operatorname{tx}}[t],\bar{\mathcal{V}}_{M-1,\operatorname{rx}}[t]}[t]=\mathbf{G}$ and $\mathbf{H}_{\bar{\mathcal{V}}_{M,\operatorname{tx}}[t],\bar{\mathcal{V}}_{M,\operatorname{rx}}[t]}[t]=(\mathbf{G}^{M-1})^{-1}$.
It is possible to construct those pairs because the resulting $\mathbf{H}_{\bar{\mathcal{V}}_{m,\operatorname{tx}}[t],\bar{\mathcal{V}}_{m,\operatorname{rx}}[t]}[t]$ is always invertible\footnote{We ignore the instances having all zeros, which give zero rate.}.
Because $\rank(\mathbf{H}_{\bar{\mathcal{V}}_{m,\operatorname{tx}}[t],\bar{\mathcal{V}}_{m,\operatorname{rx}}[t]}[t])=\rank(\mathbf{H}_{\mathcal{V}_{m,\operatorname{tx}}[t],\mathcal{V}_{m,\operatorname{rx}}[t]}[t])$, there is no rate loss by using $(\bar{\mathcal{V}}_{m,\operatorname{tx}}[t],\bar{\mathcal{V}}_{m,\operatorname{rx}}[t])$ instead of using $(\mathcal{V}_{m,\operatorname{tx}}[t],\mathcal{V}_{m,\operatorname{rx}}[t])$.

The following lemma shows the probability distribution of $\mathbf{H}_{\bar{\mathcal{V}}_{m,\operatorname{tx}}[t],\bar{\mathcal{V}}_{m,\operatorname{rx}}[t]}[t]$, which will be used to derive the achievable rate region of general multihop networks.
\begin{lemma} \label{LEM:mapping2}
Suppose a linear finite-field relay network with $p_{j,i,m}=p$ for all $i$, $j$, and $m$.
If the network has a minimum-dimensional bottleneck-hop, then for $\rank(\mathbf{G})=r\neq 0$, we obtain
\begin{eqnarray}
\!\!\!\!\!\!\!\!\!\!\!\!\!\!\!&&\Pr(\mathbf{H}_{\bar{\mathcal{V}}_{m,\operatorname{tx}}[t],\bar{\mathcal{V}}_{m,\operatorname{rx}}[t]}[t]=\mathbf{G})\nonumber\\
\!\!\!\!\!\!\!\!\!\!\!\!\!\!\!&&=\sum_{\underset{\mathcal{V}(r,r,\mathcal{V}_{m_0},\mathcal{V}_{m_0+1})}{(\mathcal{V}',\mathcal{V}'')\in}}\sum_{\mathbf{H}\in\mathcal{H}^F_{\mathcal{V}_{m_0},\mathcal{V}_{m_0+1}}(\mathbf{G},\mathcal{V}',\mathcal{V}'')}\frac{\Pr(\mathbf{H})}{|\mathcal{V}(\mathbf{H})|},
\label{EQ:p_G}
\end{eqnarray}
where $\Pr(\mathbf{H})$ is given by (\ref{EQ:p_h}).
If $p=1/2$, we obtain
\begin{equation}
\Pr(\mathbf{G})=2^{-K_{m_0+1}K_{m_0}}\frac{N_{K_{m_0+1}, K_{m_0}}(r)}{N_{r,r}(r)},
\label{EQ:p_G_equal_prob}
\end{equation}
where $N_{a,b}(i)$ is the number of instances in $\mathbb{F}_2^{a\times b}$ having rank $i$.
\end{lemma}
\begin{proof}
We refer readers to the full paper \cite{Jeon:09}.
\end{proof}

\subsubsection{Encoding, relaying, and decoding functions}
Define $\mathcal{T}_m(\mathbf{G}, \mathcal{V}'_m,\mathcal{V}'_{m+1})$ as the set of time indices of the sub-block at the $m$-th hop satisfying $\bar{\mathcal{V}}_{\operatorname{tx},m}[t]=\mathcal{V}'_m$, $\bar{\mathcal{V}}_{\operatorname{rx},m}[t]=\mathcal{V}'_{m+1}$, and $\mathbf{H}_{\mathcal{V}'_m,\mathcal{V}'_{m+1}}[t]=\mathbf{G}$, where $m\in\{1,\cdots,M\}$.
We further define
\begin{equation}
n(\mathbf{G})=n_BR\min\left\{\Pr(\mathbf{G}),\Pr\left((\mathbf{G}^{M-1})^{-1}\right)\right\}c^{-1},
\label{EQ:n_G}
\end{equation}
where
\begin{equation}
c=\frac{1}{K}\sum_{i=1}^{K_{\operatorname{min}}}i\!\!\!\!\sum_{\underset{\rank(\mathbf{G})= i}{\mathbf{G}\in\mathbb{F}^{i\times i}_2,}}\!\!\!\min\{\Pr(\mathbf{G}),\Pr((\mathbf{G}^{M-1})^{-1})\}.
\label{EQ:c}
\end{equation}

Each source will transmit $\frac{1}{K}\sum_{\mathbf{G}}\rank(\mathbf{G})n(\mathbf{G})$ bits during $n_B$ channel uses.
From (\ref{EQ:n_G}) and (\ref{EQ:c}), we can check that $R$ is equal to $\frac{1}{Kn_B}\sum_{\mathbf{G}}\rank(\mathbf{G})n(\mathbf{G})$.
For all full-rank matrices $\mathbf{G}\in\cup_{i=1}^{K_{\operatorname{min}}} \mathbb{F}_2^{i\times i}$, the detailed encoding and relaying are as follows, where $r=\rank(\mathbf{G})$.
\begin{itemize}
\item (Encoding)

For all $(\mathcal{V}'_1,\mathcal{V}'_2)\in \mathcal{V}(r,r,\mathcal{V}_1,\mathcal{V}_2)$,
declare an error if $|\mathcal{T}_1(\mathbf{G},\mathcal{V}'_1,\mathcal{V}'_2)|<n(\mathbf{G})/\big(\binom{K_1}{r}\binom{K_2}{r}\big)$, otherwise each source in $\mathcal{V}'_1$ transmits $n(\mathbf{G})/\big(\binom{K_1}{r}\binom{K_2}{r}\big)$ information bits using the time indices in $\mathcal{T}_1(\mathbf{G},\mathcal{V}'_1,\mathcal{V}'_2)$ to the nodes in $\mathcal{V}'_2$.
\item (Relaying for $m\in\{2,\cdots,M-1\}$)

For all $(\mathcal{V}'_m,\mathcal{V}'_{m+1})\in \mathcal{V}(r,r,\mathcal{V}_m,\mathcal{V}_{m+1})$,
declare an error if $|\mathcal{T}_m(\mathbf{G},\mathcal{V}'_m,\mathcal{V}'_{m+1})|<n(\mathbf{G})/\big(\binom{K_m}{r}\binom{K_{m+1}}{r}\big)$, otherwise each node in $\mathcal{V}'_m$ relays $n(\mathbf{G})/\big(\binom{K_m}{r}\binom{K_{m+1}}{r}\big)$ bits using the time indices in $\mathcal{T}_m(\mathbf{G},\mathcal{V}'_m,\mathcal{V}'_{m+1})$ to the nodes in $\mathcal{V}'_{m+1}$.
\item (Relaying for $m=M$)

For all $(\mathcal{V}'_M,\mathcal{V}'_{M+1})\in \mathcal{V}(r,r,\mathcal{V}_M,\mathcal{V}_{M+1})$,
declare an error if $|\mathcal{T}_M((\mathbf{G}^{M-1})^{-1},\mathcal{V}'_M,\mathcal{V}'_{M+1})|<n(\mathbf{G})/\big(\binom{K_M}{r}\binom{K_{M+1}}{r}\big)$, otherwise each node in $\mathcal{V}'_M$ relays $n(\mathbf{G})/\big(\binom{K_M}{r}\binom{K_{M+1}}{r}\big)$ bits using the time indices in $\mathcal{T}_M((\mathbf{G}^{M-1})^{-1},\mathcal{V}'_M,\mathcal{V}'_{M+1})$ to the destinations in $\mathcal{V}'_{M+1}$.
\end{itemize}

For simplicity, we assume that $n(\mathbf{G})/\big(\binom{K_m}{r}\binom{K_{m+1}}{r}\big)$ is an integer.
Notice that the decoding error does not occur if there is no encoding and relaying error since each destination can receive $2^{n_BR}$ information bits without interference for this case.
\subsubsection{Relaying of the received bits}
Let us now consider how each relay distributes its received bits to the nodes in the next layer.
For given $\mathbf{G}$ and $\mathcal{V}'_m$, each node in $\mathcal{V}'_m$ receives $n(\mathbf{G})/\binom{K_m}{r}$ bits and then transmits $n(\mathbf{G})/\big(\binom{K_m}{r}\binom{K_{m+1}}{r}\big)$ received bits to the nodes in $\mathcal{V}'_{m+1}$, where $r=\rank(\mathbf{G})$.
Since there exist $\binom{K_{m+1}}{r}$ possible $\mathcal{V}'_{m+1}$'s, the total number of received bits is the same as the total number of transmit bits at each node in $\mathcal{V}'_m$.
For $m\in\{2,\cdots,M-1\}$, we distribute the received bits that arrive from different paths evenly to form $n(\mathbf{G})/\big(\binom{K_m}{r}\binom{K_{m+1}}{r}\big)$ bits.
Then, among $n(\mathbf{G})/\binom{K_m}{r}$ received bits, $n(\mathbf{G})/\big(\binom{K_1}{r}\binom{K_m}{r}\big)$ received bits originate from the sources in $\mathcal{V}'_1$.
Therefore, at the last hop, the nodes in $\mathcal{V}'_{M+1}$, which are the corresponding destinations of the sources in $\mathcal{V}'_1$, can collect all received bits that originate from the sources in $\mathcal{V}'_1$.
This is because, for each node in $\mathcal{V}'_{M}$, the number of the received bits originated from $\mathcal{V}'_1$ is the same as the number of bits able to transmit to $\mathcal{V}'_{M+1}$, which are given by $n(\mathbf{G})/\big(\binom{K_1}{r}\binom{K_{M}}{r}\big)$ and $n(\mathbf{G})/\big(\binom{K_M}{r}\binom{K_{M+1}}{r}\big)$ respectively, where we use the fact that $K=K_1=K_{M+1}$.


\section{Achievable Rate Region}
In this section, we derive an achievable rate region by applying the proposed block Markov encoding and relaying.

Let $E_m$ denote the encoding error at the $m$-th hop.
Then $E_m$ occurs if
\begin{equation}
|\mathcal{T}_m(\mathbf{G},\mathcal{V}'_m,\mathcal{V}'_{m+1})|<\frac{n(\mathbf{G})}{\binom{K_m}{\rank(\mathbf{G})}\binom{K_{m+1}}{\rank(\mathbf{G})}}
\end{equation}
for any $\mathcal{V}'_m$, $\mathcal{V}'_{m+1}$, and $\mathbf{G}$, where $m\in\{1,\cdots,M-1\}$.
Similarly, $E_M$ occurs if
\begin{equation}
|\mathcal{T}_M(\mathbf{G},\mathcal{V}'_M,\mathcal{V}'_{M+1})|<\frac{n\left((\mathbf{G}^{M-1})^{-1}\right)}{\binom{K_M}{\rank(\mathbf{G})}\binom{K_{M+1}}{\rank(\mathbf{G})}}
\end{equation}
for any $\mathcal{V}'_M$, $\mathcal{V}'_{M+1}$, and $\mathbf{G}$.
Since decoding error does not occur if there is no encoding and relaying error, from the union bound, we obtain $P^{(n_B)}_{e,k}\leq\sum_{m=1}^M\Pr(E_m)$.

\begin{theorem} \label{THM:achievable_rate_multi}
Suppose a linear finite-field relay network with $M\geq2$ and $p_{j,i,m}=p$ for all $i$, $j$, and $m$.
If the network has a minimum-dimensional bottleneck-hop, then for any $\delta>0$,
\begin{equation}
\!R_k=\frac{1}{K}\sum_{i=1}^{K_{\operatorname{min}}}i\!\!\!\!\sum_{\underset{\rank(\mathbf{G})= i}{\mathbf{G}\in\mathbb{F}^{i\times i}_2,}}\!\!\!\min\{\Pr(\mathbf{G}),\Pr((\mathbf{G}^{M-1})^{-1})\}-\delta
\label{EQ:sum_rate}
\end{equation}
is achievable for all $k\in\{1,\cdots,K\}$, where $\Pr(\mathbf{G})$ is given by (\ref{EQ:p_G}).
\end{theorem}
\begin{proof}
Let us first consider $|\mathcal{T}_m(\mathbf{G},\mathcal{V}'_m,\mathcal{V}'_{m+1})|$, where $r=\rank(\mathbf{G})$.
By the weak law of large numbers \cite{Csiszar:81}, there exists a sequence $\epsilon_{n_B}\to 0$ as $n_B\to\infty$ such that the probability
\begin{equation}
|\mathcal{T}_m(\mathbf{G},\mathcal{V}'_m,\mathcal{V}'_{m+1})|\geq n_B(\Pr(\mathbf{G},\mathcal{V}'_m,\mathcal{V}'_{m+1})+\delta_{n_B})
\end{equation}
for all $\mathbf{G}$, $\mathcal{V}'_m$, and $\mathcal{V}'_{m+1}$ is greater than or equal to $1-\epsilon_{n_B}$, where $\delta_{n_B}\to 0$ as $n_B\to\infty$.
This indicates that $\Pr(E_m)\leq \epsilon_{n_B}$ if
\begin{equation}
n(\mathbf{G})\leq n_B\left(\Pr(\mathbf{G})-\binom{K_m}{r}\binom{K_{m+1}}{r}\delta_{n_B}\right)
\label{EQ:condition_m}
\end{equation}
for all $\mathbf{G}$, where $m\in\{1,\cdots,M-1\}$.
Note that we use the fact that $\Pr(\mathbf{G},\mathcal{V}'_m,\mathcal{V}'_{m+1})=\Pr(\mathbf{G})/\big(\binom{K_m}{r}\binom{K_{m+1}}{r}\big)$.
Similarly, $\Pr(E_M)\leq \epsilon_{n_B}$ if
\begin{equation}
n(\mathbf{G})\leq n_B\left(\Pr\left((\mathbf{G}^{M-1})^{-1}\right)-\binom{K_M}{r}\binom{K_{M+1}}{r}\delta_{n_B}\right)
\label{EQ:condition_M}
\end{equation}
for all $\mathbf{G}$.
Then, $P^{(n_B)}_{e,k}\leq M\epsilon_n$ if (\ref{EQ:condition_m}) and (\ref{EQ:condition_M}) hold for all $\mathbf{G}$.
This condition is satisfied if
\begin{equation}
R\leq c-\frac{c(K_{\max}!)^2\delta_{n_B}}{\min\{\Pr(\mathbf{G}),\Pr((\mathbf{G}^{M-1})^{-1})\}}
\end{equation}
for all $\mathbf{G}$, where we use the definition of $n(\mathbf{G})$ in (\ref{EQ:n_G}) and the fact that $\binom{K_m}{r}\leq K_{\max}!$.
Thus we set $R=c-\delta^*_{n_B}$, where $\delta^*_{n_B}=\frac{c(K_{\max}!)^2\delta_{n_B}}{\min_{\mathbf{G}}\{\Pr(\mathbf{G})\}}$, which tends to zero as $n_B\to\infty$.
In conclusion, we obtain
\begin{equation}
R=\frac{1}{K}\sum_{i=1}^{K_{\operatorname{min}}}i\!\!\!\!\!\sum_{\underset{\rank(\mathbf{G})= i}{\mathbf{G}\in\mathbb{F}^{i\times i}_2,}}\!\!\!\min\{\Pr(\mathbf{G}),\Pr((\mathbf{G}^{M-1})^{-1})\}-\delta^*_{n_B}
\label{EQ:achievable_rate_multi}
\end{equation}
is achievable.
Since $R$ is the symmetric rate and $\delta^*_{n_B}\to0$ as $n_B\to\infty$, this proves the assertion.
\end{proof}


Now let us consider the capacity achieving case.
If $\Pr(\mathbf{G})=\Pr((\mathbf{G}^{M-1})^{-1})$ for all possible $\mathbf{G}$, then the achievable sum-rate in Theorem \ref{THM:achievable_rate_multi} will coincide with the upper bound in (\ref{EQ:converse_multi}).
When the channel instances are uniformly distributed, the above condition holds and, as a result, the sum capacity is characterized.
The following corollary shows that the sum capacity is given by the average rank of the channel matrix of the bottleneck-hop when $p=1/2$.

\begin{corollary} \label{CO:capacity_multi}
Suppose a linear finite-field relay network with $M\geq2$ and $p_{j,i,m}=1/2$ for all $i$, $j$, and $m$.
If the network has a minimum-dimensional bottleneck-hop, the sum capacity is given by
\begin{equation}
C_{\operatorname{sum}}=2^{-K_{m_0+1}K_{m_0}}\!\!\!\!\!\sum_{\mathbf{H}\in\mathbb{F}_2^{K_{m_0+1}\times K_{m_0}}}\rank(\mathbf{H}).
\label{EQ:C_sum2}
\end{equation}
\end{corollary}
\begin{proof}
Consider the case $p_{j,i,m}=1/2$.
Then, from (\ref{EQ:p_G_equal_prob}), $\Pr(\mathbf{G})$ is a function of $\rank(\mathbf{G})$.
Since $\mathbf{G}$ and $(\mathbf{G}^{M-1})^{-1}$ are full-rank matrices, $\Pr(\mathbf{G})=\Pr((\mathbf{G}^{M-1})^{-1})$ for all possible $\mathbf{G}$.
Hence (\ref{EQ:achievable_rate_multi}) is given by
\begin{eqnarray}
R\!\!\!\!\!\!\!\!\!&&=\frac{1}{K}\sum_{i=1}^{K_{\operatorname{min}}}i\!\!\!\sum_{\underset{\rank(\mathbf{G})= i}{\mathbf{G}\in\mathbb{F}^{i\times i}_2,}}\!\!\Pr(\mathbf{G})-\delta^*_{n_B}\nonumber\\
&&=\frac{1}{K}2^{-K_{m_0+1}K_{m_0}}\sum_{i=1}^{K_{\operatorname{min}}}iN_{K_{m_0+1},K_{m_0}}(i)-\delta^*_{n_B}\nonumber\\
&&=\frac{1}{K}2^{-K_{m_0+1}K_{m_0}}\!\!\!\!\!\!\sum_{\mathbf{H}\in\mathbb{F}_2^{K_{m_0+1}\times K_{m_0}}}\!\!\!\!\!\rank(\mathbf{H})-\delta^*_{n_B},
\end{eqnarray}
where the second equality holds from (\ref{EQ:p_G_equal_prob}) and the fact that $\sum_{\mathbf{G}\in\mathbb{F}^{i\times i}_2,\rank(\mathbf{G})= i}1=N_{i,i}(i)$ and the third equality holds since $K_{\operatorname{min}}=\min\{K_{m_0},K_{m_0+1}\}$.
Since $\delta^*_{n_B}\to 0$ as $n_B\to\infty$ the achievable sum-rate $KR$ asymptotically coincides with the upper bound in (\ref{EQ:converse_multi}), which completes the proof.
\end{proof}

\begin{remark}
Theorem \ref{THM:cut_set} and Corollary \ref{CO:capacity_multi} can be extended to the $q$-ary case where inputs, outputs, and channels are in $\mathbb{F}_q$.
\label{RE:q_ary}
\end{remark}


\begin{center}
Acknowledgement
\end{center}

This work was in part supported by MKE / IITA under the program ``Next generation tactical information and communication network."


\begin{thebibliography}{10}
\providecommand{\url}[1]{#1}
\def\UrlFont{\rmfamily}
\providecommand{\newblock}{\relax} \providecommand{\bibinfo}[2]{#2}
\providecommand\BIBentrySTDinterwordspacing{\spaceskip=0pt\relax}
\providecommand\BIBentryALTinterwordstretchfactor{4}
\providecommand\BIBentryALTinterwordspacing{\spaceskip=\fontdimen2\font
plus \BIBentryALTinterwordstretchfactor\fontdimen3\font minus
  \fontdimen4\font\relax}
\providecommand\BIBforeignlanguage[2]{{%
\expandafter\ifx\csname l@#1\endcsname\relax
\typeout{** WARNING: IEEEtran.bst: No hyphenation pattern has been}%
\typeout{** loaded for the language `#1'. Using the pattern for}%
\typeout{** the default language instead.}%
\else \language=\csname l@#1\endcsname \fi #2}}


\bibitem{Etkin:08}
R. H. Etkin, D. Tse, and H. Wang, ``Gaussian interference channel capacity to within one bit,'' \emph{{IEEE} Trans. Inf. Theory}, vol. 54, pp. 5534--5562, Dec. 2008.

\bibitem{Niranjan:06}
N. Ratnakar and G. Kramer, ``The multicast capacity of deterministic relay networks with no interference,'' \emph{{IEEE} Trans. Inf. Theory}, vol. 52, pp. 2425--2432, June 2006.

\bibitem{Amir:06}
A. F. Dana, R. Gowaikar, R. Palanki, B. Hassibi, and M. Effros, ``Capacity of wireless erasure networks,'' \emph{{IEEE} Trans. Inf. Theory}, vol. 52, pp. 789--804, Mar. 2006.

\bibitem{Smith:07}
B. Smith and S. Vishwanath, ``Unicast transmission over multiple access erasure networks: Capacity and duality,'' in \emph{Proc. IEEE Information Theory Workshop}, Lake Tahoe, CA, Sept. 2007.

\bibitem{AvestimehrDiggaviTse:07}
A. S. Avestimehr, S. N. Diggavi, and D. Tse, ``Wireless network information flow,'' in \emph{Proc. 45th Annu. Allerton Conf. Communication, Control, and Computing}, Monticello, IL, Sept. 2007.

\bibitem{Mohajer:08}
S. Mohajer, S. N. Diggavi, C. Fragouli, and D. Tse, ``Transmission techniques for relay-interference networks,'' in \emph{Proc. 46th Annu. Allerton Conf. Communication, Control, and Computing}, Monticello, IL, Sept. 2008.

\bibitem{Bhadra:06}
S. Bhadra, P. Gupta, and S. Shakkottai, ``On network coding for interference networks,'' in \emph{Proc. {IEEE} Int. Symp. Information Theory}, Seattle, WA, July 2006.

\bibitem{Viveck1:08}
V. R. Cadambe and S. A. Jafar, ``Interference alignment and degrees of freedom of the K-user interference channel,'' \emph{{IEEE} Trans. Inf. Theory}, vol. 54, pp. 3425--3441, Aug. 2008.

\bibitem{Nazer:09}
B. Nazer, M. Gastpar, S. A. Jafer, and S. Vishwanath, ``Ergodic interference alignment,'' in \emph{arXiv:cs.IT/0901.4379}, Jan. 2009

\bibitem{JeonITA:09}
S.-W. Jeon and S.-Y. Chung, ``Capacity of a class of multi-source relay networks,'' in \emph{Information Theory and Applications Workshop}, San Diego, CA, Feb. 2009.

\bibitem{Jeon:09}
-------, ``Capacity of a class of multi-source relay networks,'' in preparation.

\bibitem{Csiszar:81}
I. Csisz{\'a}r and J. K{\"o}rner, \emph{Information Theory: Coding Theorems for Discrete Memoryless Systems.} New York: Academic Press, 1981.


\end{thebibliography}

\end{document}